\documentclass[11pt]{article}

\usepackage{graphicx}
\usepackage[misc,geometry]{ifsym}
\usepackage{amsmath,amssymb,amsfonts,mathtools,amsthm}
\usepackage{color}
\usepackage{hyperref}

\DeclarePairedDelimiter\floor{\lfloor}{\rfloor}

\usepackage{natbib}
\setcitestyle{round}

\newtheorem{theorem}{Theorem}
\newtheorem{proposition}{Proposition}
\newtheorem{remark}{Remark}
\newtheorem{definition}{Definition}
\newtheorem{assumption}{Assumption}

\begin{document}

\title{Diversity-Weighted Portfolios with \\Negative Parameter\thanks{The authors are greatly indebted to Adrian Banner, Robert Fernholz, Thibaut Lienart, Michael Monoyios, Vassilios Papathanakos, Julen Rotaetxe, Johannes  Ruf, and an anonymous referee for their very careful reading of the manuscript  and their many, insightful, and helpful suggestions.}
}

\author{Alexander Vervuurt  \thanks{
              Mathematical Institute, University of Oxford, Andrew Wiles Building, Radcliffe Observatory Quarter, Woodstock Road, Oxford, OX2 6GG, UK  (E-Mail: {\tt vervuurt@maths.ox.ac.uk}).
             Alexander Vervuurt gratefully acknowledges Ph.D. studentships from the Engineering and Physical Sciences Research Council, Nomura and the Oxford-Man Institute of Quantitative Finance.}
                     \and
        Ioannis Karatzas \thanks{Department of Mathematics, Columbia University, New York, NY 10027, USA (E-Mail: {\tt ik@math.columbia.edu}), and \textsc{Intech} Investment Management, One Palmer Square, Suite 441, Princeton, NJ 08542, USA (E-Mail: {\tt ik@enhanced.com}). Research supported in part by the National Science Foundation under Grant NSF-DMS-14-05210.}  
}

\date{\today}

\maketitle

\begin{abstract}
We analyze a negative-parameter variant of the diversity-weighted portfolio studied by \citet*[Finance Stoch 9(1):1--27,][]{fkk05}, which invests in each company a fraction of wealth  inversely proportional to the company's \emph{market weight} (the ratio of its capitalization to that of  the entire market).   We show that this strategy outperforms the market with probability one over sufficiently long time-horizons, under a non-degeneracy assumption on the volatility structure and under the assumption that   the market weights admit a positive lower bound. Several modifications of this portfolio are put forward, which outperform the market under milder versions of the latter \emph{no-failure} condition, and one of which is rank-based. An empirical study suggests that such strategies as studied here have  indeed the potential to outperform the market    and to be  preferable investment opportunities,  even under realistic proportional transaction costs.
\medskip
\medskip
\medskip
\medskip
\medskip


\end{abstract}

\section{Introduction}
\label{sec:intro}

Stochastic Portfolio Theory offers a relatively novel approach to portfolio selection in stock markets, aiming -- among other things -- to construct portfolios which outperform an index, or benchmark portfolio,  over a given time-horizon with probability one, whenever this might be possible. The reader is referred to R. Fernholz's monograph \citep{f02}, and to the more recent overview by \cite{fk09}. In \cite{fkk05} such outperforming portfolios have been shown to exist over sufficiently long time-horizons, in market models which satisfy certain assumptions; namely, those of weak diversity and non-degeneracy. 

One such investment strategy is the so-called \emph{diversity-weighted portfolio}. This re-calibrates the weights of the market portfolio, by raising them all to some given power $p\in(0,1)$ and then re-normalizing; see \cite{fkk05}. We study here a variant of this   strategy, namely, a diversity-weighted portfolio with \emph{negative} parameter $p<0$. This strategy invests in each company a proportion of wealth   that is \emph{inversely} proportional to the ratio 
of the company's  capitalization to the capitalization of the entire market (this ratio is called the company's \emph{market weight}). As a result, the strategy sells    the company's stock as its value increases, and buys it as its value decreases.

\subsection{Preview}

We set up the model and introduce the necessary definitions in Section \ref{sec:prelim}. Our first main result is that this negative-parameter diversity-weighted portfolio also outperforms the market over long time-horizons, but now under a \emph{no-failure} condition. This postulates that all market weights are uniformly bounded from below by a positive constant --- see Section \ref{sec:main}. The no-failure condition is  stronger  than diversity. 

A rank-based modification of this portfolio, which only invests in small-capitalization stocks, is then shown in Section \ref{sec:rank} to outperform the market, also under this assumption of ``no-failure". For certain positive parameters $p\in(0,1)$, this rank-based portfolio outperforms the market under a milder and more realistic condition, namely that of \emph{limited-failure}. This condition posits a uniform lower bound on market weights    only   for the $m<n$ largest firms by capitalization, with $n$ the total number of equities in the market.

We present two additional results in Section \ref{sec:further}. The first   shows that the negative-parameter diversity-weighted portfolio outperforms its positive-para-meter version under an additional boundedness assumption on the covariation  structure of the market;  
this posits that the eigenvalues of the covariance matrix in the underlying It\^o model are bounded from above, uniformly in time. Our second result studies a composition of these two diversity-weighted portfolios  and shows that, under the condition of diversity, this ``mix" almost surely performs better than the market over sufficiently long time-horizons.

We carry out an empirical study of the constructed portfolios in Section \ref{sec:empiric}, demonstrating that our adjustments of the diversity-weighted portfolio would have outperformed quite considerably the S\&P 500 index, if implemented on the index constituents over the 25 year period   between   January 1, 1990 and    December 31,  2014. In this study  we incorporate 0.5\% proportional transaction costs, delistings due to bankruptcies, mergers and acquisitions, and distributions such as dividends. We compute the relative returns of our portfolios over this period, as well as the Sharpe ratios relative to the market index.  
We discuss the results and possible future work in Section \ref{sec:discussion}.

\section{Preliminaries}
\label{sec:prelim}

\subsection{The Model}
\label{sec:model}
Our market model is the standard one of Stochastic Portfolio Theory \citep{f02}, where the stock capitalizations are modeled as It\^o processes. Namely, the dynamics of the $n$ stock capitalization processes $X_i(\cdot),\,  i=1,\ldots,n\,$ are given by
\begin{align} \label{model}
\mathrm{d}X_i(t) &= X_i(t)\bigg(b_i(t)\,\mathrm{d}t + \sum_{\nu=1}^d \sigma_{i\nu}(t)\, \mathrm{d}W_\nu(t)  \bigg),\quad  t\geq0,\quad i=1,\ldots,n\,;
\end{align}
here $W_1(\cdot),\ldots,W_d(\cdot)$ are independent standard Brownian motions with $\, d \ge n$, and $X_i(0)>0,\ i=1,\ldots,n$ are the initial capitalizations. We assume all processes to be defined on a probability space $(\Omega,\mathcal{F},\mathbb{P})$, and adapted to a filtration $\mathbb{F} = \{ \mathcal{F} (t)\}_{0 \le t < \infty}\,$ that satisfies the usual conditions and contains the filtration generated by the  ``driving" Brownian motions. 

\smallskip
The processes of  \emph{rates of return}  $b_i(\cdot),\ i=1,\ldots,n\,$ and  of \emph{volatilities}  $\sigma(\cdot)=(\sigma_{i\nu}(\cdot))_{1\leq i\leq n,1\leq \nu\leq d}\,,$ are $\mathbb{F}$-progressively measurable and assumed to satisfy the integrability condition
\begin{equation}
\sum_{i=1}^n \int_0^T \Big( |b_i(t)| + \sum_{\nu=1}^d(\sigma_{i\nu}(t))^2 \Big)\,\mathrm{d}t<\infty,\quad \mathbb{P}\text{-a.s.};\quad\forall \,\,T\in(0,\infty),
\end{equation}
as well as the \emph{non-degeneracy} condition
\begin{equation} \label{ND} \tag{ND}
\exists \, \varepsilon>0\ \text{ such that:}\ \, \xi'\sigma(t)\sigma'(t)\xi\geq \varepsilon ||\xi||^2,\quad \forall \,\, \xi\in\mathbb{R}^n,\ \,\,t\geq0\,; \quad\mathbb{P}\text{-a.s.}
\end{equation}

\subsection{Relative Arbitrage}
\label{sec:ra}

We study investments in the equity market described by \eqref{model} using \emph{portfolios}. These are $\mathbb{R}^n$-valued and $\mathbb{F}$-progressively measurable processes $\pi(\cdot) = \big( \pi_1 (\cdot), \cdots, \pi_n (\cdot) \big)'$, where $\pi_i(t)$ stands for the proportion of   wealth invested in stock $i$ at time $t$. 

We restrict ourselves to \emph{long-only} portfolios. These invest solely in the stocks, namely 
\begin{equation}
\pi_i(t)\geq0,\quad i=1,\ldots,n\,, \qquad \text{and} \qquad \sum_{i=1}^n \pi_i(t)=1,\quad \forall\,\, t\geq0\,;
\end{equation}
in particular,  there is no money market. The corresponding wealth process $V^\pi(\cdot)$ of an investor implementing $\pi(\cdot)$ is seen to evolve as follows (we normalize the initial wealth to 1):
\begin{equation}
\frac{\mathrm{d}V^\pi(t)}{V^\pi(t)}\,=\,\sum_{i=1}^n \pi_i(t)\, \frac{\mathrm{d}X_i(t)}{X_i(t)}\,,\qquad V^\pi(0)=1.
\end{equation}

We shall   measure performance, for the most part,  with respect to the market index. This is the wealth process $V^\mu(\cdot)$ that results from a buy-and-hold portfolio,  given by the vector process $\mu(\cdot) = \big( \mu_1(\cdot), \cdots, \mu_n (\cdot) \big)'$   of \emph{market weights}
\begin{equation} \label{defmu}
\mu_i(t)\,:=\,\frac{X_i(t)}{X(t)},\quad i=1,\ldots,n\,,\quad 
\text{where}\ \ \ X(t)\,:=\, \sum_{i=1}^n X_i(t).
\end{equation}

\begin{definition} 
\label{def:arbitrage}
A \emph{relative arbitrage} with respect to a portfolio $\rho(\cdot)$ over the time-horizon $[0,T]$, for a real number $T>0\,$, is   a portfolio $\pi(\cdot)$ such that
\begin{equation}
\mathbb{P} \big(V^\pi(T) \geq V^\rho(T)\big)=1 
\qquad \text{and}\qquad \mathbb{P}\big(V^\pi(T) > V^\rho(T)\big)>0.
\end{equation}
An equivalent way to express this notion, is to say that {\it the portfolio $\pi(\cdot)$ outperforms portfolio} $\rho(\cdot)$ over the time-horizon $[0,T]$. 

\smallskip
We call this  relative arbitrage \emph{strong}, if in fact $\mathbb{P}\big(V^\pi(T) > V^\rho(T)\big)=1\,$; and sometimes we express this by saying that the portfolio $\pi(\cdot)$   outperforms the portfolio $\rho(\cdot)$ {\it strongly} over the time-horizon $[0,T]$. \qed
\end{definition} 

We introduce the $\mathbb{R}^{n\times n}$-valued \emph{covariation} process $a(\cdot)=\sigma(\cdot)\sigma'(\cdot)$ and, writing $e_i$ for the $i$\textsuperscript{th} unit vector in $\mathbb{R}^n$, the \emph{relative covariances} 
\begin{equation} 
\tau^\mu_{ij}(t) := \big(\mu(t)-e_i\big)'a(t)\big(\mu(t)-e_j\big),\qquad 1\leq i,j\leq n.
\end{equation} 
Finally, we define the \emph{excess growth rate} $\gamma^*_\pi(\cdot)$ of a portfolio $\pi(\cdot)$ as
\begin{equation}\label{defexcess}
\gamma^*_\pi(t):=\frac{1}{2}\bigg(\sum_{i=1}^n \pi_i(t)a_{ii}(t) - \sum_{i,j=1}^n \pi_i(t) a_{ij}(t) \pi_j(t) \bigg).
\end{equation}
We shall use the reverse-order-statistics notation, defined recursively by
\begin{align}\label{deforder}
\theta_{(1)}(t)&=\max_{1\leq i \leq n} \{ \theta_i(t) \} \\
\theta_{(k)}(t)&=\max \big( \{\theta_1(t),\ldots,\theta_n(t)  \}\setminus\{ \theta_{(1)}(t),\ldots,\theta_{(k-1)}(t) \} \big), \quad k=2,\ldots,n \nonumber
\end{align}
for any $\mathbb{R}^n$-valued process $\theta(\cdot)$. Here ties are resolved lexicographically, always in favor of the lowest index $\,i$\,.  
Thus we have
\begin{equation} 
 \theta_{(1)}(t)\geq \theta_{(2)}(t)\geq \ldots \geq \theta_{(n)}(t).
\end{equation}

The non-degeneracy condition \eqref{ND} implies, on the strength of Lemma 3.4 in \cite{fk09} (originally proved in the Appendix of \cite{fkk05}), that for any long-only portfolio $\pi(\cdot)$  we have with probability one:
\begin{equation}\label{NDgamma}
\gamma^*_\pi(t) \geq \frac{\varepsilon}{\,2\,}\big(1-\pi_{(1)}(t)\big),\quad \forall \,\, t  \geq0\,.
\end{equation}

\subsection{Functionally-Generated Portfolios}
\label{sec:master}

A particular class of portfolios, called \emph{functionally-generated portfolios}, was introduced and studied by \cite{f95}. 

Consider a function $\mathbf{G}\in C^2(U,\mathbb{R}_+)$, where $U$ is an open neighborhood of 
\begin{equation}
\Delta^n_+ = \big\{ x\in\mathbb{R}^n :  x_1+\ldots+x_n=1,\ 0<x_i<1,\ i=1,\ldots,n \big\},
\end{equation}
and such that $x\mapsto x_i\mathrm{D}_i\log \mathbf{G}(x)$ is bounded on $\, \Delta^n_+ \,$ for $i=1,\ldots,n$.  
Then $\mathbf{G}$ is said to be the \emph{generating function} of the portfolio $\pi(\cdot)$ given, for $i=1,\ldots,n\,$, by
\begin{equation} \label{fgp}
\pi_i(t):=\Big( \mathrm{D}_i\log\mathbf{G}(\mu(t)) + 1 - \sum_{j=1}^n \mu_j(t)\mathrm{D}_j \log\mathbf{G}(\mu(t)) \Big)\cdot \mu_i(t).
\end{equation}
Here and throughout the paper, we write $\mathrm{D}_i$ for the partial derivative with respect to the $i^\text{th}$ variable, and $\mathrm{D}^2_{ij}$ for the second partial derivative with respect to the $i^\text{th}$ and $j^\text{th}$ variables. 
Theorem 3.1 of \cite{f95} asserts that the performance of the wealth process corresponding to $\pi(\cdot)$, when measured relative to the market, satisfies the decomposition (often referred to as ``Fernholz's master equation")
\begin{equation} \label{master} 
\log\left(\frac{V^\pi(T)}{V^\mu(T)} \right) =\log \left(\frac{\textbf{G}(\mu(T))}{\textbf{G}(\mu(0))} \right) + \int_0^T \mathfrak{g}(t)\,\mathrm{d}t\,,\quad \mathbb{P}\text{-a.s.}
\end{equation}
The quantity 
\begin{equation} \label{drift}
\mathfrak{g}(t) := \frac{-1}{2\textnormal{\textbf{G}}(\mu(t))}\sum_{i,j=1}^n \mathrm{D}_{ij}^2 \textnormal{\textbf{G}}(\mu(t))\mu_i(t)\mu_j(t)\tau^\mu_{ij}(t)
\end{equation} 
is called the \emph{drift process} of the portfolio $\pi(\cdot)$.

\subsection{Diversity-Weighted Portfolios}
\label{sec:dwp}
We consider now the \emph{diversity-weighted portfolio} with parameter $p\in\mathbb{R}$, defined as in (4.4) of \cite{fkk05}:
\begin{equation} \label{defdwp}
\pi^{(p)}_i(t)\,:=\,\frac{(\mu_i(t))^p}{\, \sum_{j=1}^n (\mu_j(t))^p\,}\,,\quad i=1,\ldots,n.
\end{equation}
One condition that was shown  to be sufficient for the existence of relative arbitrage in the model \eqref{model}, under the condition \eqref{ND}, is that of \emph{diversity} (see Corollary 2.3.5 and Example 3.3.3  
of \cite{f02}). 
This posits that no single company's capitalization can take up more than a certain proportion of the entire market. 

More formally, a market is said to be \emph{diverse} over the time-horizon $[0,T]$, for some real number $T>0$, if
\begin{equation} \label{D} \tag{D}
\exists \, \delta \in(0,1)\ \text{ such that: }\ \ \mathbb{P} \big( \mu_{(1)}(t)<1-\delta\, ,\ \ \forall \,t\in [0,T]\big) \,=\,1\,.
\end{equation}
Models of the form \eqref{model}, which satisfy the property \eqref{D} and admit local martingale deflators, were explicitly shown to exist in Remark 6.2 of \cite{fkk05}. Other constructions have been proposed by \cite{or06} and \cite{sar14}. 
\citet[see Eq.~(4.5)]{fkk05} showed that the portfolio \eqref{defdwp} outperforms the market index $\mu(\cdot)$ strongly, in markets satisfying \eqref{ND} and \eqref{D}, for any $p\in(0,1)$, and over time-horizons $[0,T]$ with 
\begin{equation} \label{oldT}
T\, \in \, \left( \frac{\,2\log n\,}{\varepsilon \delta p}\,, \,\infty\right).
\end{equation}

In Theorem \ref{thm:dwpneg} below, we show that a similar property holds for the diversity-weighted portfolio with \emph{negative} parameter $p$, in markets with the following \emph{no-failure} condition:
\begin{equation} \label{NF} \tag{NF}
\exists \,\varphi \in (0,1/n)\ \text{ such that }\,\,\,\mathbb{P}\big( \mu_{(n)}(t)> \varphi,\ \ \forall \,t\in[0,T]\big)\,=\,1\,.
\end{equation}
We note that this condition implies diversity with parameter $\delta = (n-1)\varphi\,.$

\section{Main Result}
\label{sec:main}

Here is the main result of this paper. 

\begin{theorem} 
\label{thm:dwpneg}
In the market model \eqref{model}, and under the assumptions \eqref{ND} and \eqref{NF}, the diversity-weighted portfolio $\pi^{(p)}(\cdot)$ with parameter 
\begin{equation} 
\label{range}
p \,\in \,\left( \frac{\log n}{\, \log (n\varphi)\,}\,,\,0\right)
\end{equation}
 is a strong  arbitrage relative to the market $\mu(\cdot)$ over the time-horizon  $[0,T]$, for any real number   
\begin{equation} \label{dwphorizon}
T \,>\, \frac{-2n \log (n \varphi)}{\, \varepsilon (1-p) \big(n-(n \varphi)^p \big) \,}\,.
\end{equation}
\end{theorem}
\begin{proof}
We apply the theory of functionally-generated portfolios as in Section \ref{sec:master}. For any $p\neq 0$, it is checked that the portfolio \eqref{defdwp} is generated by the function
\begin{equation} 
\label{dwpgf}
\mathbf{G}_p\,:x\,\longmapsto \, \left(\sum_{i=1}^n x_i^{\,p}\right)^{1/p}\!.
\end{equation}
We apply the method of Lagrange multipliers to maximize this function $\mathbf{G}_p$ for $p<0$ subject to $\sum_{i=1}^n x_i=1$; this gives $x_i=1/n,\ i=1,\ldots,n.$ We may therefore write that, under \eqref{NF}, the generating function admits the lower and upper bounds 
\begin{equation} 
\label{gpbound}
n^{1-p}=\sum_{i=1}^n \bigg(\frac{1}{n}\bigg)^p \leq \sum_{i=1}^n \big(\mu_i(t)\big)^p = \big(\mathbf{G}_p(\mu(t)\big)^p < \sum_{i=1}^n \varphi^p = n\varphi^p.
\end{equation}
Hence, for any $T\in (0, \infty)$, we have the lower bound
\begin{equation} \label{boubound}
\log \left(\frac{\mathbf{G}_p(\mu(T))}{\mathbf{G}_p(\mu(0))} \right) > \log (n\varphi),
\end{equation}
a negative number.
Using assumption \eqref{NF} and the lower bound from \eqref{gpbound}, we obtain 
\begin{equation} 
\label{ineq}
\pi^{(p)}_{(1)}(t)=\frac{(\mu_{(n)}(t))^p}{\,\sum_{i=1}^n (\mu_i(t))^p\,} < \frac{\varphi^p}{n^{1-p}}=\frac{(n\varphi)^p}{n}<1,
\end{equation}
where the last (strict) inequality follows from (\ref{range}). 
In conjunction  with \eqref{ineq}, the inequality \eqref{NDgamma}   gives 
\begin{equation} 
\label{lowerbd}
\int_0^T \gamma^*_{\pi^{(p)}}(t)\,\mathrm{d}t\, \geq \, \frac{\,\varepsilon\,}{2}\int_0^T\big(1-\pi^{(p)}_{(1)}(t)\big)\,\mathrm{d}t \, > \,  \frac{\,\varepsilon\,}{2}\,T \left(1-\frac{(n\varphi)^p}{n} \right).
\end{equation}
Finally, straightforward computation shows that the drift process of the portfolio \eqref{defdwp}, as defined in \eqref{drift}, is equal for all $p\in\mathbb{R}$   
to
\begin{equation} \label{dwpdrift}
\mathfrak{g}_p(\cdot)=(1-p)\, \gamma^*_{\pi^{(p)}}(\cdot).
\end{equation}

 We apply now  \eqref{dwpdrift}, \eqref{boubound} and \eqref{lowerbd} to Fernholz's master equation \eqref{master}, and  conclude that the relative performance of $\pi^{(p)}(\cdot)$ over $[0,T]$, with respect to the market, is given by
\begin{align} \label{masterdwp}
\log \bigg(\frac{V^{\pi^{(p)}}\!(T)}{V^\mu(T)} \bigg) &=\log \left(\frac{\mathbf{G}_p(\mu(T))}{\mathbf{G}_p(\mu(0))} \right) + (1-p)\int_0^T \gamma^*_{\pi^{(p)}}(t) \,\mathrm{d}t\\
&>\log (n\varphi) + (1-p)\,  \frac{\,\varepsilon\,}{2}\,T \left(1-\frac{(n\varphi)^p}{n} \right)\nonumber \\
& > 0,\qquad \mathbb{P}\text{-a.s.},\nonumber
\end{align}
provided   $T$ satisfies \eqref{dwphorizon}.
Hence the portfolio $\pi^{(p)}(\cdot)$  outperforms the market strongly  over sufficiently long time-horizons $[0,T]$, as indicated in \eqref{dwphorizon}.  
\end{proof}

\section{Rank-Based Variants}
\label{sec:rank}

Real markets typically do \emph{not} satisfy 
the \eqref{NF} assumption (as stocks can crash), so it is desirable to weaken this condition    of Theorem \ref{thm:dwpneg}. One possible modification of \eqref{NF}, is to posit that ``no-failure" only holds for the $\mathfrak{m}$ stocks ranked highest by capitalization, for some fixed $2<\mathfrak{m}<n$, namely
\begin{equation}
\label{LF} 
\tag{LF}
\exists \, \kappa \in (0,1/\mathfrak{m})    
\ \text{ such that }\ \,\mathbb{P} \big( \,\mu_{(\mathfrak{m})}(t)> \kappa\,,  \,\, \forall \,\,t\in[0,T]\,\big) \,=\,1\,.
\end{equation}

We call this condition 
\eqref{LF} the \emph{limited-failure} condition, as it postulates that no more than $n-\mathfrak{m}$ companies will ``go bankrupt" by time $T$. In this context, bankruptcy of company $i$ during the time-horizon $[0,T]$ is defined as the event $\{ \exists \, \,t\in[0,T]$ such that $\mu_i(t)\leq \kappa\}$.   
 We  also note that \eqref{LF} implies the diversity condition with parameter $(\mathfrak{m}-1)\kappa\,.$   
 
 \begin{remark}
 As was noted by V. Papathanakos (private communication), the diversity condition \eqref{D} with parameter $\delta\in(0,1)$ in turn implies \eqref{LF} with parameters 
\begin{equation} \label{Vassilios}
\mathfrak{m}=\floor*{\frac{1}{1-\delta}}\quad \text{ and }\quad \kappa = \frac{1-(\mathfrak{m}-1)(1-\delta)}{n-(\mathfrak{m}-1)}\, <\, \frac{1}{\frak{m}}\, ,
\end{equation} 
with $\floor*{x}$ the largest integer less than or equal to $x\in\mathbb{R}$. To see this, we note that, as long as $(k-1)(1-\delta)<1$, for $1<k\leq n$, we have the implication
\begin{equation}\label{Vassilios2}
\mu_{(k-1)}(t)<1-\delta\quad \Rightarrow\quad \mu_{(k)}(t) > \frac{1-(k-1)(1-\delta)}{n-(k-1)}\, .
\end{equation}
We identify the inequality on the right hand side of \eqref{Vassilios2} as a version of \eqref{LF}. The largest integer for which this lower bound is positive, is $k=\floor*{1/(1-\delta)}$, giving \eqref{Vassilios}. \qed
 \end{remark}

\subsection{A Diversity-Weighted Portfolio of Large Stocks}

Under the assumption \eqref{LF}, one can attempt to construct a relative arbitrage using a variant (introduced in Example 4.2 of \cite{f01})  of the diversity-weighted portfolio with parameter $p=r\in\mathbb{R}$ 
 which only invests in the $m=\mathfrak{m}\,$ highest-ranked stocks (and thus naturally avoids investing in ``crashing" stocks), namely:
\begin{equation} 
\label{largedwp}
\pi^\#_{p_t(k)}(t)\,=\,\begin{dcases} 
\frac{(\mu_{(k)}(t))^r}{\,\sum_{\ell=1}^m (\mu_{(\ell)}(t))^r\,}\, , &k=1,\ldots,m,\\
0, &k=m+1,\ldots,n.
\end{dcases}
\end{equation}
Here, $p_t(k)$ is the index of the stock   ranked $k$\textsuperscript{th} at time $t$ (with ties  resolved ``lexicographically" once again,  by choosing the lowest index), so that $\mu_{p_t(k)}(t)=\mu_{(k)}(t)$.

\smallskip
We shall denote by $\mathfrak{L}^{k,k+1}(t)\equiv \Lambda^{\Xi_k}(t)$   the semimartingale local time  accumulated {\it at the origin,}  over the time-interval $[0,t]$,  by the continuous, non-negative semimartingale
\begin{equation}
\Xi_k(\cdot)\,:=\,\log\big(\mu_{(k)}(\cdot)/\mu_{(k+1)}(\cdot)\big),\qquad k=1,\ldots,n-1.
\end{equation}

\begin{definition} 
Let us consider the local times at the origin of all  continuous, nonnegative semimartingales of the form 
$$
\log\big(\mu_{(k)}(\cdot)/\mu_{(k+r)}(\cdot)\big)\,, \qquad k=1, \ldots, n-r\,, \,\,\, r \ge 2\,,
$$
and call them \emph{``higher-order collision local times"}.   \qed
 \end{definition} 

We shall assume throughout this section the following:

\begin{assumption} 
\label{assume:evan}   
All   higher-order collision local times  vanish.   \qed
 \end{assumption}

The reader should consult \cite{bg08} for general theory on ranked semimartingales, as well as \cite{ipbkf11} for sufficient conditions that ensure the evanescence of such higher-order collision local times, as posited in Assumption \ref{assume:evan}.  

\smallskip
When attempting to construct a relative arbitrage with the portfolio \eqref{largedwp},  one encounters a problem. To wit: an application of Theorem 3.1 of \cite{f01} asserts that the following master equation holds for this rank-based portfolio:
\begin{align} 
\label{rankmaster}
\log \bigg(\frac{V^{\pi^\#}\!(T)}{V^\mu(T)}\bigg) = \log \left(\frac{\mathcal{G}^\#_r(\mu(T))}{\mathcal{G}^\#_r(\mu(0))}\right) &+ (1-r)\int_0^T\gamma^*_{\pi^\#}(t)\,\mathrm{d}t\\ &- \int_0^T \frac{\pi^\#_{p_t(m)}(t)}{2} \, \mathrm{d}\mathfrak{L}^{m,m+1}(t),\nonumber
\end{align}
with $\mathcal{G}^\#_r$ the generating function of $\pi^\#(\cdot)$; compare with \eqref{masterdwp}. Due to the unbounded nature of semimartingale  local time, the final term in \eqref{rankmaster} (referred to as ``leakage" by \cite{f01}) admits no obvious almost sure bound. Thus there is no lower bound on the market-relative performance of $\pi^\#(\cdot)$ that holds under reasonable conditions. 

Since in real markets this local time term is typically small, we do still expect this portfolio to have a good performance, and therefore include it in our empirical study --- see Section \ref{sec:empiric}.

\subsection{A Diversity-Weighted Portfolio of Small Stocks}

Let us then fix $1\leq m<n$, and define a \emph{small-stock} diversity-weighted portfolio by   
\begin{equation}
 \label{smalldwp}
\pi^\flat_{p_t(k)}(t)\,=\,\begin{dcases} 
0, &k=1,\ldots,m,\\
\frac{(\mu_{(k)}(t))^r}{\sum_{\ell=m+1}^n (\mu_{(\ell)}(t))^r}\,, &k=m+1,\ldots,n.
\end{dcases}
\end{equation}
With this new dispensation,  the local time term in \eqref{rankmaster} changes sign for $\pi^\flat(\cdot)$ (see equation \eqref{rankmastersmall} below) and the   problems mentioned in the previous subsection disappear. In fact,  we can show that the new portfolio in \eqref{smalldwp} outperforms the market under the assumptions \eqref{ND}, \eqref{LF}, for certain positive values of the parameter $r\,$; as well as under the assumptions \eqref{ND},  \eqref{NF},  when $r$ is within an appropriate  range of negative values.

\begin{proposition} 
\label{prop:rankdwpneg}
We place ourselves in the context of the market model \eqref{model}, and under   Assumption \ref{assume:evan} and the condition  \eqref{ND}.  

\medskip
\noindent
 \textnormal{(i)} 
 If the condition \eqref{LF} also holds, the small-stock diversity-weighted portfolio of \eqref{smalldwp} with parameter 
\begin{equation}
\label{eq: Range1}
r\in \left(-\,\frac{\log(2)}{\,\log((m+1)\kappa)\,}\,,\,1\right) \qquad \text{and} \qquad m=\mathfrak{m}-2
\end{equation}
is a strong relative arbitrage with respect to the market  $\mu(\cdot)$ over the time-horizon  $[0,T]$, provided that  
\begin{equation} 
\label{rankdwphorizon+}
\kappa<\frac{1}{2(m+1)}  \quad \text{and} \quad
T> T^\flat_+ :=  \frac{\,4 \left[ \,\log\left(\frac{n-m}{2}\right)-r\log \big( (m+1) \kappa \big)\right]\,}{\varepsilon r(1-r) \big(2-  (m+1)^{-r}\kappa^{-r} \big)}\, .
\end{equation}
\textnormal{(ii)} If  \eqref{NF} also holds,   the small-stock diversity-weighted portfolio of \eqref{smalldwp} with parameter  
\begin{equation}
\label{eq: Range2}
r\in\left(\frac{\, \log(n-m)\,}{\,\log( (m+1) \varphi )\,}\,,\,0\right)
\end{equation}
 is a strong  arbitrage relative to the market $\mu(\cdot)$ over the time-horizon $[0,T]$, provided  
\begin{equation}
 \label{rankdwphorizon-}
T\, >\, T^\flat_- \,:=\,\frac{-2(n-m)\log((m+1) \varphi )}{\, \varepsilon (1-r) \big(n-m-(m+1)^r\varphi^r\big)\,}\,.
\end{equation}
\end{proposition}
\begin{proof} 
The portfolio \eqref{smalldwp} is generated by the function 
\begin{equation}
\mathcal{G}_r(x)\,:=\,\left(\,\sum_{\ell=m+1}^n \big(x_{(\ell)}\big)^r \right)^{1/r}\!,
\end{equation} 
a rank-based variant of \eqref{dwpgf}. In a manner analogous to \eqref{rankmaster}, Theorem 3.1 of \cite{f01} implies the following master equation for the performance of our small-stock portfolio:
\begin{align} 
\label{rankmastersmall}
\log \bigg(\frac{V^{\pi^\flat}\!(T)}{V^\mu(T)}\bigg) \, = \, \log \left(\frac{\mathcal{G}_r(\mu(T))}{\mathcal{G}_r(\mu(0))}\right) &+ (1-r)\int_0^T\gamma^*_{\pi^\flat}(t)\,\mathrm{d}t\\ &+ \int_0^T \frac{\pi^\flat_{p_t(m)}(t)}{2} \, \mathrm{d}\mathfrak{L}^{m,m+1}(t).\nonumber
\end{align}
We note the change of sign  for the last term, when juxtaposed with the analogue \eqref{rankmaster} of this equation for the large-stock portfolio in \eqref{largedwp}.

\paragraph{Case (i):} With  $r\in \big(-\log(2)/\log((m+1)\kappa),1\big)\,$ as in (\ref{eq: Range1}), we have  
\begin{equation} \label{r+below}
\kappa<\frac{1}{2(m+1)}\ \ \ \ \Longrightarrow\ \ \ \ r> -\frac{\log(2)}{\log\big((m+1)\kappa\big)}>0\,.
\end{equation}
Recalling (LF) and $\, m = \mathfrak{m} -2\,$, we have the following bounds 
\begin{align} 
\label{grbound+}
2\kappa^r &<  (\mu_{(\mathfrak{m}-1)}(t))^r+(\mu_{(\mathfrak{m})}(t))^r+ \sum_{\ell=\mathfrak{m}+1}^n (\mu_{(\ell)}(t))^r\nonumber\\
&= \sum_{\ell=m+1}^n (\mu_{(\ell)}(t))^r = \big(\mathcal{G}_r(\mu(t))\big)^r\\
&\leq \sum_{\ell=m+1}^n \bigg( \frac{1}{\, m+1\,}\bigg)^r = (n-m)\,\big( m+1 \big)^{-r} \nonumber
\end{align}
for the generating function $\mathcal{G}_r\,$, and  lead to the lower bound 
\begin{equation} 
\label{bouboundr+}
\log \left(\frac{\mathcal{G}_r(\mu(T))}{\mathcal{G}_r(\mu(0))}\right) > \log \big( (m+1)\kappa \big) -\frac{1}{\,r\,}\log\bigg(\frac{n-m}{2}\bigg).
\end{equation}
Note that, using the lower bound in \eqref{grbound+} and 
$\mu_{(k)}(t)\leq 1/k\,$, $\,k=1,\ldots,n\,$, we obtain
\begin{equation} 
\label{pirbound+}
\pi^\flat_{(1)}(t)=\frac{(\mu_{(m+1)}(t))^r}{\sum_{\ell=m+1}^n (\mu_{(\ell)}(t))^r} < \frac{(m+1)^{-r}}{2\kappa^r} < 1,
\end{equation}
where the last bound follows from the inequalities in \eqref{r+below}.   

\medskip
The non-degeneracy condition \eqref{ND} now gives
\begin{equation} \label{lowerbdr+}
\int_0^T \gamma^*_{\pi^\flat}(t)\,\mathrm{d}t \, \geq \, \frac{\varepsilon}{2}\int_0^T \big(1-\pi^\flat_{(1)}(t)\big)\,\mathrm{d}t \,> \,\frac{\,\varepsilon\,}{2} \,T\bigg(1-\frac{1}{2(m+1)^r\kappa^r}\bigg) 
\end{equation}
 in conjunction with  \eqref{NDgamma} and \eqref{pirbound+}. Whereas the nonnegativity of $\pi^\flat_{p_t(m)}( t)\,,$ coupled with the nondecrease of the local time  $\mathfrak{L}^{m,m+1}(\cdot) ,$  shows that 
\begin{equation} \label{localtimebd}
\int_0^T \frac{\pi^\flat_{p_t(m)}(t)}{2}\,\mathrm{d}\mathfrak{L}^{m,m+1}(t)\,\geq \, 0\,.
\end{equation}
We use this and apply \eqref{bouboundr+} and \eqref{lowerbdr+} to the rank-based master equation \eqref{rankmastersmall}, to obtain
\begin{align} \label{masterrankdwp+}
\log \bigg(\frac{V^{\pi^\flat}\!(T)}{V^\mu(T)}\bigg) &> \log \big( (m+1)\kappa \big)-\frac{1}{\,r\,}\log\bigg(\frac{n-m}{2}\bigg)  \\
& \quad +\frac{\, \varepsilon \,}{2}(1-r)\,T\bigg(1-\frac{1}{2(m+1)^r\kappa^r}\bigg) \nonumber\\
&> 0\,,\qquad \mathbb{P}\text{-a.s.},\nonumber
\end{align}
if and only if $T > T^\flat_+$ as defined in \eqref{rankdwphorizon+}. We conclude that, under these conditions, $\pi^\flat(\cdot)$ strongly outperforms the market over the horizon  $[0,T]$.   

\paragraph{Case (ii):} With $r\in \big( \log(n-m)/\log(n\kappa),0\big)$ and 
under the condition \eqref{NF}, we have the following bounds
\begin{equation}
 \label{grbound-}
(n-m)\, (m+1)^{-r} \leq  \sum_{\ell=m+1}^n (\mu_{(\ell)}(t))^r = \big(\mathcal{G}_r(\mu(t))\big)^r < (n-m)\,\varphi^r\,.
\end{equation}
For the first inequality, we have used the simple fact that $\,\mu_{(\ell)}(t) \le \mu_{(m+1)}(t) \le 1 / (m+1)\,$ holds for $\, \ell = m+1, \cdots, n\,$. 
Whereas, from (\ref{eq: Range2}),  we have
\begin{equation} \label{pirbound-}
\pi^\flat_{(1)}(t)=\frac{(\mu_{(n)}(t))^r}{\sum_{\ell=m+1}^n (\mu_{(\ell)}(t))^r} < \frac{\varphi^r}{\,(n-m)\, (m+1)^{-r}\,}< 1 
\end{equation}
by analogy with \eqref{ineq}. The inequalities \eqref{grbound-} lead to the lower bound 
\begin{equation} \label{bouboundr-}
\log \left(\frac{\mathcal{G}_r(\mu(T))}{\mathcal{G}_r(\mu(0))}\right) >  \log \big( (m+1) \varphi \big)\,;
\end{equation}
and the non-degeneracy condition \eqref{ND}, in conjunction with  \eqref{NDgamma} and the upper bound \eqref{pirbound-} on the largest portfolio weight $\pi^\flat_{(1)}(\cdot)$, give
\begin{equation} \label{lowerbdr-}
\int_0^T \gamma^*_{\pi^\flat}(t)\,\mathrm{d}t \, \geq \, \frac{\,\varepsilon\,}{2}\int_0^T \big(1-\pi^\flat_{(1)}(t)\big)\,\mathrm{d}t \, >\, \frac{\,\varepsilon\,T\,}{2}  \left(1-\frac{\, (m+1)^r \varphi^r\,}{\,n-m\,}\right).
\end{equation}
Again, we use \eqref{localtimebd} and apply \eqref{bouboundr-} and \eqref{lowerbdr-} to the master equation \eqref{rankmastersmall}, to obtain
\begin{align} 
\label{masterrankdwp-}
\log \bigg(\frac{V^{\pi^\flat}\!(T)}{V^\mu(T)}\bigg) &> \log \big( (m+1) \varphi \big) + \frac{\,\varepsilon\,T\,}{2}(1-r) \left(1-\frac{\, (m+1)^r \varphi^r\,}{n-m}\right) \\
&> 0,\qquad \mathbb{P}\text{-a.s.},\nonumber
\end{align}
provided $T > T^\flat_-$ as defined in \eqref{rankdwphorizon-}. 

Hence this small-stock, negative-parameter portfolio $\pi^\flat(\cdot)$, outperforms the market over the horizon  $[0,T]$.
\end{proof}




\section{Further Considerations}
\label{sec:further}

We  present now some other results, the first of which shows that the diversity-weighted portfolio \eqref{defdwp} with $p\in(0,1)$ is outperformed by its negative-para-meter counterpart under sufficient conditions. We also study a particular combination of these two types of diversity-weighted portfolios, showing that it outperforms a non-degenerate diverse market --- see Proposition \ref{dwpmix} further below.

\subsection{Negative vs. Positive Parameter}

In the following Proposition \ref{dwpposneg}, besides \eqref{ND} we will impose also an \emph{upper} bound on the eigenvalues of the covariance matrix $a(\cdot)$, in the form of the \emph{bounded variance} assumption:
\begin{equation} \label{BV} \tag{BV}
\exists \,K>0\ \text{ such that:}\ \ \xi'\sigma(t)\sigma'(t)\xi \leq K ||\xi||^2,\quad \forall \, \xi\in\mathbb{R}^n,\ \,t\geq0\,; \quad\mathbb{P}\text{-a.s.}
\end{equation} 
By Lemma 3.5 of \cite{fk09}, the bounded variance condition \eqref{BV} implies that for long-only portfolios $\pi(\cdot)$ we have the almost sure inequality
\begin{equation} \label{BVgamma}
\gamma^*_\pi(t)\leq 2K\big(1-\pi_{(1)}(t)\big),\quad \forall \, t\geq0.
\end{equation}


\begin{proposition} 
\label{dwpposneg}
Let us place ourselves in the market model \eqref{model},   under the conditions \eqref{ND}, \eqref{BV} and \eqref{NF}.

The diversity-weighted portfolio $\pi^{(p^-)}(\cdot)$ with negative parameter $$p^-\in \left( \frac{\log n}{\,\log (n\varphi)\,}\,,\,0\right)$$ is then a strong arbitrage relative to the   diversity-weighted portfolio $\pi^{(p^+)}(\cdot)$ with positive parameter
\begin{equation}
\label{p+min}
p^+\in \bigg( \max \Big\{ 0, 1-\frac{\varepsilon (n-(n\varphi)^{p^-})(1-p^-)}{4K(n-1)} \Big\} , 1 \bigg)\, ,
\end{equation}
over any horizon $[0,T]$ of length 
\begin{equation} 
\label{dwphorizon2}
T > \frac{-2\log (n\varphi)}{C}\, .
\end{equation}
Here the positive constant $C$ is defined as
\begin{equation} \label{defC}
C\,:=\,\frac{\,\varepsilon\,}{2}\left(1-\frac{\,(n\varphi)^{p^-}\,}{n}\right) \big(1-p^-\big)-\frac{\,2K\,}{n}\, (n-1)(1-p^+)\,.
\end{equation}
\end{proposition}
\begin{proof}
To simplify notation a bit, we write $\pi^\pm(\cdot)$ and $\mathbf{G}_\pm$ for $\pi^{(p^\pm)}(\cdot)$ and $\mathbf{G}_{p^\pm}$, respectively. Note that for $p^+>0$ the inequalities in \eqref{gpbound} reverse, i.e., 
\begin{equation} \label{gposbound}
n\varphi^{p^+} < \big(\mathbf{G}_+(\mu(t))\big)^{p^+} \leq n^{1-p^+},
\end{equation}
which again gives the lower bound of \eqref{boubound} for $p=p^+$. 
Using \eqref{BVgamma} with the observation   $\pi^+_{(1)}(t)\geq 1/n$, we get that 
\begin{equation} \label{gammaplus}
\int_0^T \gamma^*_{\pi^+}(t) \,\mathrm{d}t  \, \leq  \, 2K\int_0^T \big(1-\pi^+_{(1)}(t)\big)\,\mathrm{d}t \, \leq \, 2K T\big(1-(1/n) \big).
\end{equation}
Hence, recalling \eqref{dwpdrift}, we see that by virtue of \eqref{gposbound}~
and \eqref{gammaplus} the master equation \eqref{master} for $\pi(\cdot)=\pi^+(\cdot)$ leads to the upper bound
\begin{equation} 
\label{masterdwp+}
\log \bigg(\frac{V^{\pi^+}(T)}{V^{\mu}(T)}\bigg) < -\log(n\varphi)+2(1-p^+)K T\big(1-(1/n) \big),\quad \mathbb{P}\text{-a.s.}
\end{equation}
Combining \eqref{masterdwp} and \eqref{masterdwp+}, we find that
\begin{align} 
\label{dwpcomp}
\log \bigg(\frac{V^{\pi^-}(T)}{V^{\pi^+}(T)}\bigg) &= \log \bigg(\frac{V^{\pi^-}(T)}{V^{\mu}(T)}\bigg) - \log \bigg(\frac{V^{\pi^+}(T)}{V^{\mu}(T)}\bigg)\\
&> 2\log (n\varphi) + C T\nonumber\\
&> 0,\quad \mathbb{P}\text{-a.s.},\nonumber
\end{align}
provided that
\begin{equation} \label{eqCT}
CT > -2\log (n\varphi) >0\,.
\end{equation}
Here, the constant $\,C\,$ is  given by \eqref{defC}. An easy calculation shows that \eqref{p+min} implies $C>0$, whereas the last inequality in \eqref{eqCT} comes from $\varphi<1/n$. It follows that the first inequality in \eqref{eqCT} is equivalent to the condition posited in \eqref{dwphorizon2}, and that in this case $\pi^-(\cdot)$  outperforms $\pi^+(\cdot)$ strongly over the time-horizon $[0,T]$. 
\end{proof}


 Proposition \ref{dwpposneg}\, shows that, as long as the diversity-weighted portfolio $\pi^{(p^+)}(\cdot)$ is ``sufficiently similar" to the market portfolio $\mu(\cdot)$ (and thus ``far enough" from $\pi^{(p^-)}(\cdot)$), it is outperformed  strongly, over sufficiently long time-horizons, by the diversity-weighted portfolio $\pi^{(p^-)}(\cdot)$ with negative parameter -- provided, of course,  that the aforementioned conditions on the volatility structure and non-failure of stocks hold.


\begin{remark} 
Compared to \eqref{D}, the stronger \eqref{NF} assumption implies a possibly shorter minimal time horizon than that of \eqref{oldT}, over which the diversity-weighted portfolio with parameter $p^+\in(0,1)$ is guaranteed strongly to outperform the market. Namely, the fact that $\mathbf{G}_{p^+}(x)\geq1,\ \forall x\in\Delta^n_+,$ together with \eqref{gposbound}, implies that
\begin{equation} 
\max\big\{1,n\varphi^{p^+}\big\} \leq \big(\mathbf{G}_{p^+}(\mu(t))\big)^{p^+} \leq n^{1-p^+}.
\end{equation}
Recall from the end of Section \ref{sec:dwp} that \eqref{NF} implies diversity \eqref{D} with parameter $\delta=(n-1)\varphi$. It follows from \eqref{oldT}, and using an argument similar to that in the proof of Theorem \ref{thm:dwpneg}, that the positive-parameter diversity-weighted portfolio is a strong arbitrage relative to the market over horizons $[0,T]$, with
\begin{equation}
T >  \min\bigg\{  \frac{\,2\log n\,}{\varepsilon \delta p}, \frac{-2\log \big(n\delta/(n-1)\big)}{\varepsilon \delta (1-p)} \bigg\}.
\end{equation}
This is an improvement of \eqref{oldT} if and only if $\varphi>n^{-1/p^+}$.
\qed
\end{remark}


\subsection{Mixing}

Since the no-failure assumption \eqref{NF} does not hold in real markets, it is of interest to find variants of \eqref{defdwp} which exhibit   similar performance, but require weaker assumptions. The rank-based variant in Proposition \ref{prop:rankdwpneg} is one attempt at this; Proposition \ref{dwpmix} right below is another.
\begin{proposition} \label{dwpmix}
Define the portfolio
\begin{equation} \label{mixpi}
\widehat{\pi}(t)\,=\,\mathfrak{p}(t) \pi^+(t) +  (1-\mathfrak{p}(t)) \pi^-(t),
\end{equation}
with $\pi^\pm(\cdot)=\pi^{(p^\pm)}(\cdot)$ diversity-weighted portfolios as defined in \eqref{defdwp} with $p^+\in(0,1)$ and $p^-<0$, and the mixing proportion $\mathfrak{p}(\cdot)$ given by
\begin{equation}\label{defprop}
\mathfrak{p}(t)=\frac{\mathbf{G}_{p^+}(\mu(t))}{\mathbf{G}_{p^+}(\mu(t))+\mathbf{G}_{p^-}(\mu(t))}\in(0,1). 
\end{equation}
In the market model \eqref{model}, with assumptions \eqref{ND} and \eqref{D}, the portfolio $\widehat{\pi}(\cdot)$ is a strong arbitrage relative to the market $\mu(\cdot)$, over time-horizons $[0,T]$ with
\begin{equation} \label{mixhorizon}
T \,>\, \mathfrak{T}\,:=\, \frac{2(1+n^{(1/p^-)-1}) \log\big( n^{(1/p^+)-1}+n^{(1/p^-)-1} \big)}{\varepsilon \delta(1-p^+)}\, .
\end{equation}
\end{proposition}
\begin{proof}
The portfolio \eqref{mixpi} is generated by the function $\widehat{\mathbf{G}}= \mathbf{G}_++\mathbf{G}_-$, with $\mathbf{G}_\pm=\mathbf{G}_{p^\pm}$ the generating functions of $\pi^{\pm}(\cdot)=\pi^{(p^\pm)}(\cdot)$ as in \eqref{dwpgf}; let $\mathfrak{g}_\pm(\cdot)=\mathfrak{g}_{p^\pm}(\cdot)$ be their drift processes as in \eqref{dwpdrift}. 

The drift process of the composite portfolio $\widehat{\pi}(\cdot)$ is then
\begin{align} 
\label{driftmix}
\widehat{\mathfrak{g}}(t) 
&= \frac{-1}{2\widehat{\mathbf{G}}(\mu(t))}\sum_{i,j=1}^n \mathrm{D}_{ij}^2 \widehat{\mathbf{G}}(\mu(t))\mu_i(t)\mu_j(t)\tau^\mu_{ij}(t)\nonumber\\
 &=  \mathfrak{p}(t)\, \frac{-1}{2\mathbf{G}_+(\mu(t))}\sum_{i,j=1}^n \mathrm{D}_{ij}^2 \mathbf{G}_+ (\mu(t))\mu_i(t)\mu_j(t)\tau^\mu_{ij}(t)\nonumber\\
&\quad +(1-\mathfrak{p}(t))\,\frac{-1}{2\mathbf{G}_-(\mu(t))}\sum_{i,j=1}^n \mathrm{D}_{ij}^2 \mathbf{G}_- (\mu(t))\mu_i(t)\mu_j(t)\tau^\mu_{ij}(t)\nonumber\\
&=\mathfrak{p}(t) \mathfrak{g}_+(t) +  (1-\mathfrak{p}(t)) \mathfrak{g}_-(t)\nonumber\\
&= (1-p^+) \mathfrak{p}(t) \gamma^*_{\pi^+}(t)+ (1-p^-)(1- \mathfrak{p}(t)) \gamma^*_{\pi^-}(t)\, ;
\end{align}
the final step follows from \eqref{dwpdrift}. We note that 
\begin{equation} \label{gamma0}
\gamma^*_{\pi^-}(t)\geq0,
\end{equation}
as this holds for any long-only portfolio by Lemma 3.3 of \cite{fk09}, which together with the observation that $\mathfrak{p}(t)<1$ allows us to obtain the bound
\begin{equation} \label{driftbd}
\widehat{\mathfrak{g}}(t) \geq (1-p^+) \mathfrak{p}(t) \gamma^*_{\pi^+}(t).
\end{equation}

We   
note the simple bounds
\begin{equation} \label{gp+bound}
1=\sum_{i=1}^n \mu_i(t) \leq \sum_{i=1}^n \big(\mu_i(t)\big)^{p^+} = \big(\mathbf{G}_+(\mu(t)\big)^{p^+} \leq n^{1-p^+},
\end{equation}
and that   
the lower bound in \eqref{gpbound} holds even without \eqref{NF}, so now
\begin{equation} \label{gp-bound}
0 \,\leq \,\bigg(\sum_{i=1}^n \big(\mu_i(t)\big)^{p^-}\bigg)^{1/p^-} = \,\mathbf{G}_-(\mu(t)) \, \leq \, n^{(1/p^-)-1}\,.
\end{equation}
Using the lower bound from \eqref{gp+bound}, and the upper bound from \eqref{gp-bound}, we assert that  
\begin{equation} 
\label{pbd}
\mathfrak{p}(t)= \left(1+\frac{\mathbf{G}_-(\mu(t))}{\mathbf{G}_+(\mu(t))}\right)^{-1} \geq \frac{1}{1+n^{(1/p^-)-1}}>0.
\end{equation}
One can easily see that in the positive-parameter diversity-weighted portfolio, one's proportion of wealth invested relative to the market portfolio is diminished for large-cap stocks, and increased for small-capitalization stocks; therefore
\begin{equation} \label{bdpibig}
\pi^+_{(1)}(t) = \frac{(\mu_{(1)}(t))^p}{\sum_{i=1}^n (\mu_i(t))^p} \leq \mu_{(1)}(t).
\end{equation}
The non-degeneracy condition \eqref{ND} implies \eqref{NDgamma}, which by \eqref{bdpibig} and the diversity assumption \eqref{D} leads to 
\begin{equation} \label{gammabd}
\gamma^*_{\pi^+}(t)\geq \frac{\varepsilon}{2} \big(1-\pi^+_{(1)}(t)\big)\geq \frac{\varepsilon}{2}\big(1-\mu_{(1)}(t)\big) > \frac{\varepsilon}{2}\, \delta \,.
\end{equation}

Finally, applying \eqref{driftbd} to the master equation \eqref{master}, and then using the bounds \eqref{gp+bound},  \eqref{gp-bound}  and \eqref{pbd}, \eqref{gammabd}, we conclude that
\begin{align} \label{mastermix}
\log \bigg(\frac{V^{\widehat{\pi}}(T)}{V^\mu(T)} \bigg) &\geq \log \bigg( \frac{\mathbf{G}_+(\mu(T))+\mathbf{G}_-(\mu(T))}{\mathbf{G}_+(\mu(0))+\mathbf{G}_-(\mu(0))} \bigg) + (1-p^+)\int_0^T \mathfrak{p}(t) \gamma^*_{\pi^+}(t)\,\mathrm{d}t\nonumber\\
&> - \log\big( n^{(1/p^+)-1}+n^{(1/p^-)-1} \big) + \frac{\varepsilon}{\,2\,}\,\big(1-p^+\big)\,\frac{\delta\, T}{1+n^{(1/p^-)-1}}\nonumber \\
&>0\,,\quad \mathbb{P}\text{-a.s.}
\end{align}
under the condition that $T$ satisfies \eqref{mixhorizon}.
\end{proof}

We have shown that the long-only portfolio \eqref{mixpi} outperforms strongly the market over sufficiently long time-horizons, under the assumptions of non-degeneracy and diversity. This is a property that the portfolio $\pi^+(\cdot)$ also has on its own, as proved in the Appendix of \cite{fkk05}. 

\smallskip
We remark also that, since the ``threshold" $\,\mathfrak{T}\,$ of \eqref{mixhorizon} is strictly decreasing in $p^-$, and the lower bound \eqref{mastermix} is strictly increasing in $p^-$ for fixed $T$, our result becomes stronger the closer the negative parameter $\,p^-\, $ gets to the origin. As we take $p^-\uparrow 0$, we recover the well-known result that $\widehat{\pi}(\cdot)=\pi^{(p^+)}(\cdot)$ is a strong arbitrage relative to the market. 

\begin{remark}
The statement of Proposition \ref{dwpmix} can be strengthened, by weakening the diversity condition \eqref{D} to that of \emph{weak diversity} over the horizon $[0,T]$. This notion is defined in equation (4.2) of \cite{fkk05} as 
\begin{equation} 
\label{WD} 
\tag{WD}
\exists \, \,\delta \in(0,1)\ \ \text{ such that:}\ \ \ \ \mathbb{P} \left( \frac{1}{\,T\,}\int_0^T \mu_{(1)}(t)\, \mathrm{d}t <1-\delta\right) \,=\,1\,.
\end{equation}
Under this weaker assumption, it is straightforward to see that the following modification of \eqref{gammabd} will hold, thus maintaining the validity of the proof:
\begin{equation} 
\int_0^T \gamma^*_{\pi^+}(t)\,\mathrm{d}t\geq \frac{\,\varepsilon\,}{2} \int_0^T (1-\pi^+_{(1)}(t))\,\mathrm{d}t \geq \frac{\,\varepsilon\,}{2} \int_0^T (1-\mu_{(1)}(t))\,\mathrm{d}t > \frac{\,\varepsilon\,}{2}\,\delta \,T.
\end{equation}
\qed
\end{remark}

\subsection{Arbitrary Time Horizons}

We should note here that, since both the \eqref{NF} and \eqref{LF} conditions imply diversity, by Lemma 8.1 and Example 8.3 of \cite{fkk05} it follows that \emph{short-term} relative arbitrage exists. That is, using the ``mirror portfolios" of Fernholz, Karatzas and Kardaras, one can construct a portfolio which outperforms the market over {\it any} time-horizon $[0,T]$.  

\subsection{Two Portfolios with a Threshold}
We mention briefly  two other variants of the portfolio \eqref{defdwp} with $p<0$, which exhibit similar behavior for mid- and upper-range market weights (in the sense that they invest in stocks which decrease in value relative to the market, and sell stock when a company's market weight increases), but start selling stock once a firm's market weight falls below a certain threshold. The idea behind this threshold is that it represents a value below which the investor fears bankruptcy of the firm, and wishes to liquidate the position in the firm so as to minimize   losses. The two portfolios we propose have weights
\begin{equation} 
\label{pigamma}
  \Gamma_i(t) = \frac{\mu_i(t)^{k}e^{-\mu_i(t)/\theta}}{\sum_{j=1}^n \mu_j(t)^{k}e^{-\mu_j(t)/\theta}}\,, \qquad i=1, \ldots, n
\end{equation}
and
\begin{equation} \label{pibeta}
B_i(t) = \frac{\mu_i(t)^{\alpha} (1-\mu_i(t))^{\beta}}{\sum_{j=1}^n \mu_j(t)^\alpha (1-\mu_j(t))^\beta}\,, \qquad i=1, \ldots, n\,.
\end{equation}
Here, $k,\,\theta$ and $\alpha,\,\beta$ are all positive constants; the threshold values below which the portfolio weights become \emph{increasing} functions of market weights 
 are $\theta k$ and $ \alpha/(\alpha+\beta)$, respectively. So far, our results regarding these portfolios remain restricted to empirical ones (see Section \ref{sec:empiric}). A theoretical result for the $\Gamma(\cdot)$ portfolio might follow in a way similar to that of Proposition 1 in \cite{bf08}.

\section{Empirical Results}
\label{sec:empiric}
We test the validity and applicability of our theoretical results by conducting an empirical study of the performance of our portfolios using historical market data.
More precisely, we back-test by investing according to the portfolios studied here  in the $n=500$ daily constituents of the S\&P 500 index for $T=6301$ consecutive trading days between 1 January 1990 and 31 December 2014. We obtain these  data from the Compustat and CRSP data sets  (these data sets are obtainable from the website\ \url{http://wrds-web.wharton.upenn.edu/wrds/}).  
 We have incorporated dividends and delistings such as mergers, acquisitions and liquidations. 

\medskip
We aim to simulate as realistically as possible the wealth evolution of an investor implementing our portfolios without any additional information, and as such impose on each trade proportional transaction costs equal to 0.5\% of the total absolute value traded. We rebalance the portfolio using a simple Total Variance criterion, that is, we only rebalance when the total variance ``distance"
\begin{equation} \label{TVdef}
\mathbf{TV}(\widetilde{\pi}(t),\pi(t))\,=\,\sum_{i=1}^n \,\widetilde{\pi}_i(t)  \, \big| \widetilde{\pi}_i(t)-\pi_i(t) \big|
\end{equation}
between the target portfolio $\pi(\cdot)$ and the portfolio $\widetilde{\pi}(\cdot)$ becomes larger than a certain threshold, which we determine empirically by trial-and-error --- that is, we simply try several threshold values for each portfolio, and select the one which gives the highest return over the entire holding period for that particular strategy.  
Here, $\widetilde{\pi}(t)$ is the vector of proportions of wealth invested in stocks obtained when the investor does \emph{not} rebalance at time $t$, that is 
$$\widetilde{\pi}_i(t) =  \bar{\pi}_i(t-1)\frac{V^{\bar{\pi}}(t-1)}{V^{\bar{\pi}}(t)}\frac{X_i(t)}{X_i(t-1)}\, , \qquad i=1,\ldots,n,\quad t=2,\ldots,T,$$
where $\bar{\pi}(t-1)$ is the portfolio that was actually implemented in the previous time step.  

We use \verb;R; to program our simulations --- the code is available  
upon request. We summarize our findings in Table \ref{tab:emp}, which displays the average annual  relative returns in excess of the market (denoted ``Market-RR"), and the Sharpe ratios of our portfolios over the entire 25-year period that we use; the latter is computed as  
\begin{equation}
\text{SharpeRatio}(\pi(\cdot)) = \frac{\overline{R^\pi}}{\, \text{StdDev}(R^\pi)\,}\cdot \sqrt{\frac{T}{25}}\,. 
\end{equation}
Here, $R^\pi=\{ R^\pi(t),\, t=1,\ldots,T \}$ are the daily returns of the portfolio $\pi$, with mean $\overline{R^\pi}$, namely
\begin{equation}
R^\pi(t)=\frac{V^\pi(t)}{V^\pi(t-1)}-1,\quad t=2,\ldots,T,
\end{equation}
and the sample standard deviation $\text{StdDev}$ is defined as follows for any sequence of numbers $x_1,\ldots,x_k\in\mathbb{R}$ with average $\overline{x}$:
\begin{equation}
\big(\text{StdDev}(x_1,\ldots,x_k)\big)^2 := \frac{1}{k-1} \sum_{i=1}^k (x_i-\overline{x})^2.
\end{equation}
Moreover, we compute the following measure of performance for each portfolio:
\begin{equation}
\widetilde{\gamma}_\pi := \frac{\gamma_\pi}{\, \text{StdDev}(R^\pi)\,}\cdot \frac{1}{25}\, ,
\end{equation}
that is, the total growth rate divided by the sample standard deviation. The former is estimated as
\begin{align}
\gamma_\pi =\log\left(\frac{V^\pi(T)}{V^\pi(1)}\right).
\end{align}
Figure \ref{fig:V} showcases the wealth processes corresponding to our portfolios.

\bigskip

\begin{table} [h!]
\caption{ 
{\small
Some measures of performance for the studied portfolios when traded between 1 January 1990 and 31 December 2014 on the constituents of the S\&P 500 index, over which the market had an average annual return of 9.7190\%. We write $\pi^\text{E}(\cdot)$ for the equally-weighted portfolio $\pi^\text{E}_i(\cdot)=1/n,\ i=1,\ldots,n$. The parameter ``$\mathbf{TV}$ Threshold" determines over which value of $\mathbf{TV}$ as in \eqref{TVdef} we rebalance, which we determine in-sample by trial-and-error.
\medskip
\medskip
\medskip
}
}
\label{tab:emp}  
\begin{tabular}{llllll}
\hline\noalign{\smallskip}
Portfolio & $\mathbf{TV}$ Threshold & Market-RR & Sharpe Ratio & $\widetilde{\gamma}_\pi$ \\
\noalign{\smallskip}\hline\noalign{\smallskip}
$\mu(\cdot)$                      & N/A    & 0\% & 0.60477 & 8.1678  \\
\noalign{\smallskip}
$\pi^\text{E}(\cdot)$             & 0.0005 & 2.2331\% & 0.70428 & 9.7127 \\
\noalign{\smallskip}
$\pi^{(p)}(\cdot),\ p=0.5$        & 0.0022 & 1.6125\% & 0.76766 & 10.949 \\
\noalign{\smallskip}
$\pi^{(p)}(\cdot),\ p=-0.5$         & 0.0015  & 4.9655\% & 0.78558 & 10.910\\
\noalign{\smallskip}
$\pi^\#(\cdot),\ r=-0.5,\,m=470$    & 0.0025 & 2.5778\% & 0.76979 & 10.873 \\
\noalign{\smallskip}
$\pi^\flat(\cdot),\ r=0.5,\,m=30$ & 0.0001 & 0.9578\% & 0.64713 & 8.8215 \\
\noalign{\smallskip}
$\pi^\flat(\cdot),\ r=-0.5,\,m=30$  & 0.0100 & 1.7645\% & 0.67412 & 9.2114 \\ 
\noalign{\smallskip}
$\widehat{\pi}(\cdot),\ p^+\!=0.5,\,p^-\!=-0.5$  & 0.0022 & 1.6125\% & 0.76766  & 10.949 \\
\noalign{\smallskip}
$\Gamma(\cdot),\ k=0.65,\,\theta=1\text{e-}4$ & 0.0020 & 3.6336\% & 0.68160 & 9.0796 \\
\noalign{\smallskip}
$B(\cdot),\ \alpha=1\text{e-}4,\,\beta=2$     & 0.0002 & 1.8701\% & 0.67919 & 9.2920 \\
\noalign{\smallskip}\hline
\end{tabular}
\bigskip
\bigskip
\end{table}


\begin{figure} [h!]
\centering
\includegraphics[width=\textwidth]{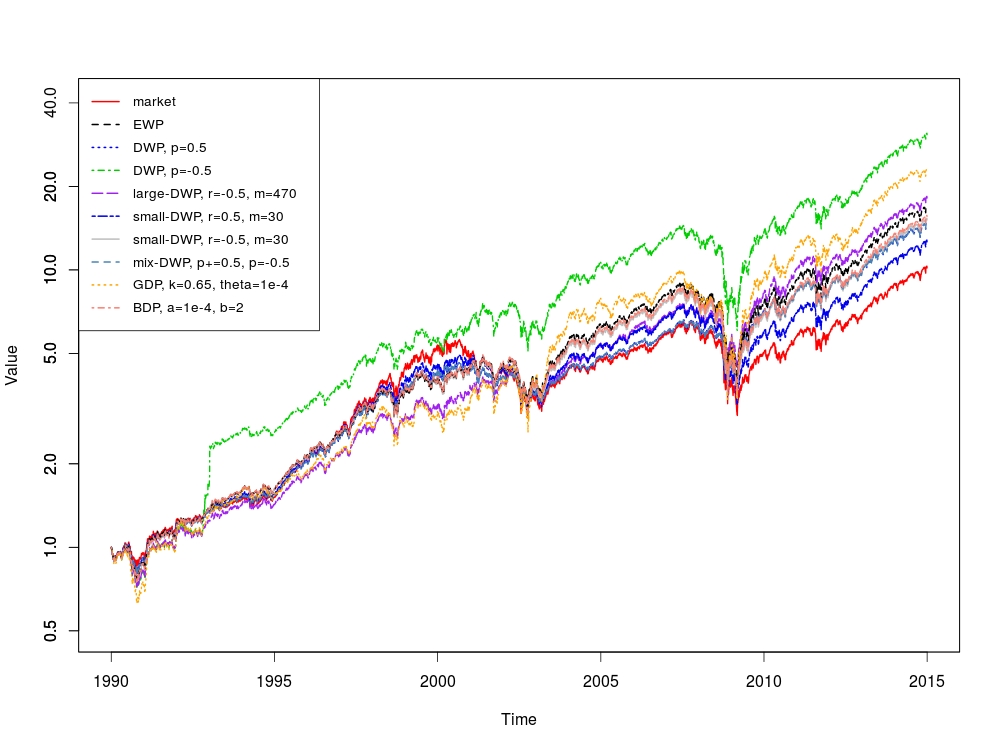}
\caption{
{\small 
The wealth processes corresponding to the portfolios in Table \ref{tab:emp}, namely: the market portfolio $\mu(\cdot)$; the diversity-weighted portfolios $\pi^{(p)}(\cdot)$ with $p=0$ (i.e.~the equally-weighted portfolio), $p=0.5$, and $p=-0.5$; the rank-based diversity-weighted portfolios $\pi^\#(\cdot)$ with $r=-0.5,\,m=470$ and $\pi^\flat(\cdot)$ with $r=0.5,\,m=30$, and $r=-0.5,\,m=30$; the mixed portfolio $\widehat{\pi}(\cdot)$ from \eqref{mixpi}, with $p^+=0.5,\, p^-=-0.5$; and the portfolios \eqref{pigamma} and \eqref{pibeta}, with $k=0.65,\,\theta=1\text{e-}4$ and $\alpha=1\text{e-}4,\,\beta=2$, respectively.}
}
\label{fig:V} 
\bigskip
\end{figure}

From Table \ref{tab:emp} and Figure \ref{fig:V}, we can see that all portfolios outperform the market by quite a margin (however, only after the year 2000), both in terms of the market-relative return, as well as the Sharpe Ratio. We also note that the simulated realizations of wealth processes are much more volatile than that of the market; the portfolios seem to exploit market growth much better, but also lose value more quickly when the market does poorly. Moreover, the negative-parameter portfolios we study appear to perform better than their positive-parameter versions over the period studied. The reader will notice that the positive-parameter portfolio $\pi^{(p)}(\cdot)$ with $p=0.5$, and the mixing portfolio $\widehat{\pi}(\cdot)$, give identical results, which is because the weight \eqref{defprop} of the positive-parameter diversity-weighted portfolio is always very near 1 (namely $1-\mathfrak{p}(t)\sim 1\text{e-}11$). Finally, we wish to stress that the above portfolio and threshold parameters were optimized only heuristically, and therefore there is considerable  room for improvement (that is, through more systematic optimization) --- the performance of the portfolios is quite sensitive to these parameters, as well as the level of transaction costs. 
Our empirical study is therefore merely a demonstration of potential outperformance of the market with the portfolios studied for a small investor.

\section{Discussion and Suggestions for Future Research}
\label{sec:discussion}

We wish to point out that all of our results require certain assumptions on the volatility structure of the market, as well as on the behavior of the market weights. The no-failure condition \eqref{NF} definitely \emph{does not} hold in real markets, since typically some companies do  crash. The weaker \eqref{LF} assumption is an improvement on this, and can be argued to hold for $m$ not too close to $n$ --- although, to the authors' knowledge, there is no mechanism holding this in place, as there is for the diversity assumption \eqref{D} which can be imposed by anti-trust regulation. In this regard,  see \cite{sf11}, as well as the  recent paper by \cite{ks14} in which the number of companies is allowed to fluctuate --  due to both splits and mergers of companies.

It has been raised by V. Papathanakos (private communication) 
that from a practitioner's point of view, there are several restrictions on the implementation of negative-parameter diversity-weighted portfolios on a large scale. One of these is that one typically demands that no position owns more than, say, 1\% of the outstanding shares of a security; our portfolios strongly invest in small stocks. Another constraint is related to liquidity: regular rebalancing in a predictable way (which is implied by all functionally-generated portfolios) incurs very high transaction costs. It would be useful to model these phenomena and their influence on portfolio performance. 

More generally, it would be of great interest to develop a theory of transaction costs in the framework of Stochastic Portfolio Theory, which would allow one to improve upon, or even optimize over, rebalancing rules given a certain portfolio. A first attempt at this was made by \cite[Section 6.3]{f02}, where R. Fernholz estimates the turnover in a diversity-weighted portfolio.  

Another idea is to replace the almost sure assumptions \eqref{D}, \eqref{LF}, and \eqref{NF}, by \emph{probabilistic versions} of these assumptions, where the corresponding bounds on market weights hold with a probability close to but smaller than 1. It would be interesting to see whether \emph{probabilistic} relative arbitrages could be constructed, which have a certain ``likelihood of outperforming" the market --- see \cite{bhs12} for a first study in this direction.  

One could also incorporate additional information on expected drifts and bankruptcies, to improve the simple portfolios set forth in this paper. The approach by \cite{palw13} might be applied to achieve this. Moreover, it would be of interest to develop  methods for finding the optimal relative arbitrage within the class of functionally-generated portfolios; an  attempt at this was made in \cite{palw14}, and for more general strategies in \cite{fernk10} and \cite{fk11}. Limitations on the existence of relative arbitrages with respect to certain portfolios have been established in the recent paper   
\cite{wong15}.

\section*{Acknowledgements}
Wharton Research Data Services (WRDS) was used in preparing the data for this paper. This service and the data available thereon constitute valuable intellectual property and trade secrets of WRDS and/or its third-party suppliers.


\bibliographystyle{plainnat}
\bibliography{sptrefs2}   

\end{document}